\newif\ifconferenceversion

\newif\iffullversion
\ifconferenceversion
\else
\fullversiontrue
\fi

\iffullversion
\documentclass{patmorin}
\pdfoutput=1
\usepackage{pat}
\else 
\documentclass{cccg12}
\newcommand{\R}[0]{\mathds{R}}
\newcommand{\Z}[0]{\mathds{Z}}
\fi

\usepackage{graphicx,amssymb,amsmath,url}

\usepackage{url}
\usepackage{subfig}
\usepackage{dsfont}

\title{Visibility-Monotonic Polygon Deflation\thanks{\ifconferenceversion
    The full version of this paper is available
    at \protect\url{http://arxiv.org/abs/XXXX.XXXX}.  \fi School of
    Computer Science, Carleton University, {\tt \{jit, vida, nhoda,
    morin\}@scs.carleton.ca}}}

\author{Prosenjit Bose\and Vida Dujmovi{\'c}\and Nima Hoda\and Pat
  Morin}

\ifconferenceversion
\index{Bose, Prosenjit}
\index{Dujmovi{\'c}, Vida}
\index{Hoda, Nima}
\index{Morin, Pat}
\fi

\begin{document}

\maketitle

\begin{abstract}
  A \emph{deflated} polygon is a polygon with no visibility crossings.
  We answer a question posed by Devadoss et al. (2012) by presenting a
  polygon that cannot be deformed via continuous visibility-decreasing
  motion into a deflated polygon.  \iffullversion We show that the
  least $n$ for which there exists such an $n$-gon is seven.  \fi In
  order to demonstrate non-deflatability, we use a new combinatorial
  structure for polygons, the directed dual, which encodes the
  visibility properties of deflated polygons.  We also show that any
  two deflated polygons with the same directed dual can be deformed,
  one into the other, through a visibility-preserving deformation.
\end{abstract}

\section{Introduction}

Much work has been done on visibilities of polygons \cite{Ghosh07,
  ORourke87} as well as on their convexification, including work on
convexification through continuous motions \cite{Connely00}.  Devadoss
et al. \cite{Devadoss09} combine these two areas in asking the
following two questions: (1) Can every polygon be convexified through
a deformation in which visibilities monotonically increase?  (2) Can
every polygon be deflated (i.e. lose all its visibility crossings)
through a deformation in which visibilities monotonically decrease?

The first of these questions was answered in the affirmative at CCCG
2011 by Aichholzer et al. \cite{Aichholzer11b}.
\ifconferenceversion
\else 
In this paper, we resolve the second question in the negative by
presenting a non-deflatable polygon, shown in
Figure~\ref{fig:counter-poly:poly}.  While it is possible to use
\textit{ad hoc} arguments to demonstrate the non-deflatability of this
polygon, we develop a combinatorial structure, the directed dual, that
allows us to prove non-deflatability for this and other examples using
only combinatorial arguments.  We also show that seven is the least
$n$ for which there exists a non-deflatable $n$-gon in general
position.

As a byproduct of developing the directed dual, we obtain the
following additional results: (1) The vertex-edge visibility graph of
a deflated polygon is completely determined by its directed dual; and
(2) any deflated polygon may be monotonically deformed into any other
deflated polygon having the same directed dual.

\fi


\section{Preliminaries}

We begin by presenting some definitions.  Here and throughout the
paper, unless qualified otherwise, we take \emph{polygon} to mean
simple polygon on the plane.

A \emph{triangulation}, $T$, of a polygon, $P$, with vertex set $V$ is
a partition of $P$ into triangles with vertices in $V$.  The
\emph{edges} of $T$ are the edges of these triangles and we call such
an edge a \emph{polygon edge} if it belongs to the polygon or, else, a
\emph{diagonal}.  A triangle of $T$ with exactly one diagonal edge is
an \emph{ear} and the \emph{helix} of an ear is its vertex not
incident to any other triangle of $T$.


Let $w$ and $uv$ be a vertex and edge, respectively, of a polygon,
$P$, such that $u$ and $v$ are seen in that order in a
counter-clockwise walk along the boundary of $P$.  Then $uv$
is \emph{facing} $w$ if $(u,v,w)$ is a left turn.  Two vertices or a
vertex and an edge of a polygon are \emph{visible} or \emph{see} each
other if there exists a closed line segment contained inside the
closed polygon joining them.  If such a segment exists that intersects
some other line segment then they are visible \emph{through} the
latter segment.  We say that a polygon is in \emph{general position}
if the open line segment joining any of its visible pairs of vertices
is contained in the open polygon.

\begin{figure}[htb]
  \centering
  \subfloat[]{\label{fig:vv-vis-graph:poly}
    \includegraphics{figs/poly.mps}}
  \quad
  \subfloat[]{\label{fig:vv-vis-graph:vg}
    \includegraphics{figs/poly-vg.mps}}
  \caption{\protect\subref{fig:vv-vis-graph:poly} A polygon
    and \protect\subref{fig:vv-vis-graph:vg} its visibility graph.}
  \label{fig:vv-vis-graph}
\end{figure}

The \emph{visibility graph} of a polygon is the geometric graph on the
plane with the same vertex set as the polygon and in which two
vertices are connected by a straight open line segment if they are
visible (e.g. see Figure~\ref{fig:vv-vis-graph}).


\subsection{Polygon Deflation}

A \emph{deformation} of a polygon, $P$, is a continuous, time-varying,
simplicity-preserving transformation of $P$.  Specifically, to each
vertex, $v$, of $P$, a deformation assigns a continuous mapping
$t\mapsto v^t$ from the closed interval $[0,1]\subset\R$ to the plane
such that $v^0 = v$.  Additionally, for $t\in[0,1]$, $P^t$ is simple,
where $P^t$ is the polygon joining the images of $t$ in these mappings
as their respective vertices are joined in $P$.

A \emph{monotonic deformation} of $P$ is one in which no two vertices
ever become visible, i.e., there do not exist $u$ and $v$ in the
vertex set of $P$ and $s, t\in [0, 1]$, with $s < t$, such that $u^t$
and $v^t$ are visible in $P^t$ but $u^s$ and $v^s$ are not visible in
$P^s$.


A polygon is \emph{deflated} if its visibility graph has no edge
intersections.  Note that a deflated polygon is in general position
and that its visibility graph is its unique triangulation.  Because of
this uniqueness and for convenience, we, at times, refer to a deflated
polygon and its triangulation interchangeably.  A \emph{deflation} of
a polygon, $P$, is a monotonic deformation $t\mapsto P^t$ of $P$ such
that $P^1$ is deflated.  If such a deformation exists, then $P$ is
\emph{deflatable}.


\subsection{Dual Trees of Polygon Triangulations}

\begin{figure}[htb]
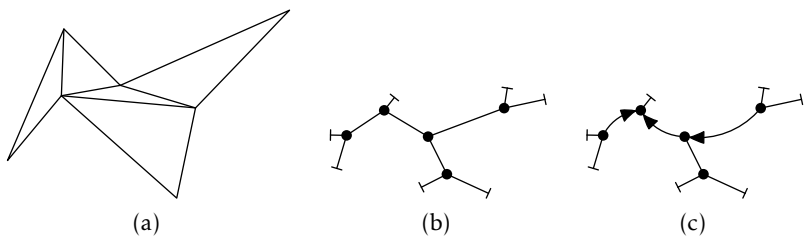

  \centering
  \subfloat[]{\label{fig:pt-dt-dd:tri}
    \includegraphics{figs/poly-tri.mps}}
  \quad
  \subfloat[]{\label{fig:pt-dt-dd:dual}
    \includegraphics{figs/poly-dual.mps}}
  \quad
  \subfloat[]{\label{fig:pt-dt-dd:dd}
    \includegraphics{figs/poly-dirdual.mps}}
  \caption{\protect\subref{fig:pt-dt-dd:tri} A polygon
    triangulation, \protect\subref{fig:pt-dt-dd:dual} its dual tree
    and \protect\subref{fig:pt-dt-dd:dd} its directed dual.  Triangle
    and terminal nodes are indicated with disks and tees,
    respectively.}
  \label{fig:pt-dt-dd}
\end{figure}

The \emph{dual tree}, $D$, of a polygon triangulation, $T$, is a plane
tree with a \emph{triangle node} for each triangle of $T$, a
\emph{terminal node} for each polygon edge of $T$ and where two nodes
are adjacent if their correspondents in $T$ share a common edge.  The
dual tree preserves edge orderings of $T$ in the following sense.  If
a triangle, $a$, of $T$ has edges $e$, $f$ and $g$ in
counter-clockwise order then the corresponding edges of its
correspondent, $a^D$, in $D$ are ordered $e^D$, $f^D$ and $g^D$ in
counter-clockwise order (e.g. see Figure~\ref{fig:pt-dt-dd:dual}).

Note that the terminal and triangle nodes of a dual tree have degrees
one and three, respectively.  We call the edges of terminal nodes
\emph{terminal edges}.


\begin{figure}[htb]
  \centering
  \includegraphics{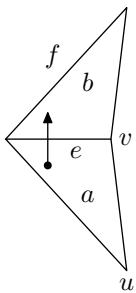}
  \caption{A pair of triangles, $a$ and $b$, sharing an edge, $e$,
    such that their quadrilateral union has a reflex vertex and a
    single segment path from $a$ to $b$ contained in the open
    quadrilateral.  The reflex endpoint, $v$, of $e$ is to the right
    of the path and so the pair $(a,b)$ is right-reflex.}
  \label{fig:rrefl}
\end{figure}

An ordered pair of adjacent triangles $(a,b)$ of a polygon
triangulation, $T$, is \emph{right-reflex} if the quadrilateral union
of $a$ and $b$ has a reflex vertex, $v$, situated on the right-hand
side of a single segment path from $a$ to $b$ contained in the open
quadrilateral.  We call $v$ the \emph{reflex endpoint} of the edge
shared by $a$ and $b$ (see Figure~\ref{fig:rrefl}).

The \emph{directed dual}, $D$, of a polygon triangulation, $T$, is a
dual tree of $T$ that is partially directed such that, for every
right-reflex pair of adjacent triangles $(a, b)$ in $T$, the edge
joining the triangle nodes of $a$ and $b$ in $D$ is directed
$a\rightarrow b$ (e.g. see Figure~\ref{fig:pt-dt-dd:dd}).  Note that
if $P$ is deflated, then for every pair of adjacent triangles, $(a,
b)$, of $T$ one of $(a, b)$ or $(b, a)$ is right-reflex and so every
non-terminal edge in $D$ is directed.

Throughout this paper, as above, we use superscripts to denote
corresponding objects in associated structures.  For example, if $a$
is a triangle of the triangulation, $T$, of a polygon and $b$ is a
triangle node in the dual tree, $D$, of $T$ then $a^D$ and $b^T$
denote the node corresponding to $a$ in $D$ and the triangle
corresponding to $b$ in $T$, respectively.


\section{Directed Duals of Deflated Polygons}

In this section, we derive some properties of deflated polygons and
use them to relate the visibilities of deflated polygons to paths in
their directed duals.  We also show that two deflated polygons with
the same directed dual can be monotonically deformed into one another.

\begin{lemma}
  \label{lem:cutear-defl}
  Let $P$ be a deflated polygon, let $a$ be an ear of $P$ and let $P'$
  be the polygon resulting from removing $a$ from $P$.  Then $P'$ is
  deflated.
\end{lemma}
\iffullversion
\begin{proof}
  $P'$ is a subset of $P$, so if a vertex pair is visible in $P'$ then
  the corresponding pair is visible in $P$.  Then a crossing in the
  visibility graph of $P'$ would imply one in that of $P$.
\end{proof}
\fi

\begin{cor}
  \label{cor:subpoly-defl}
  If the union of a subset of the triangles of a deflated polygon
  triangulation is a polygon, then it is deflated.
\end{cor}

\begin{lemma}
  \label{lem:defl-vis-thr-uniq}
  If $u$ is a vertex opposite a closed edge, $e$, in a triangle of a
  deflated polygon triangulation, then $u$ sees exactly one polygon
  edge through $e$.
\end{lemma}
\iffullversion
\begin{proof}
  If $e$ is a polygon edge then $u$ sees no other edge through $e$
  than $e$ itself.  Otherwise, if $u$ saw more than one polygon edge
  through $e$, it would also see some vertex through $e$, implying a
  visibility crossing in the visibility graph of the deflated
  polygon---a contradiction.
  
  Now, since the polygon is bounded, a sufficiently long open line
  segment starting on $u$ and intersecting $e$ must contain points
  both interior and exterior to the polygon.  Then it must intersect
  the polygon boundary and, since the polygon is deflated, the
  intersection point must be on an open polygon edge visible to $u$.
\end{proof}
\fi

\begin{figure}[htb]
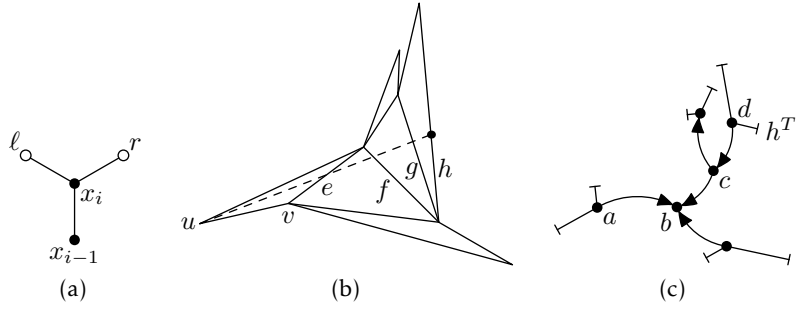

  \centering
  \subfloat[]{\label{fig:vis-path:recur}
    \includegraphics{figs/vispath-iter.mps}} \quad
  \subfloat[]{\label{fig:vis-path:indseq}
    \includegraphics{figs/defl-poly-ind-path.mps}} \quad
  \subfloat[]{\label{fig:vis-path:eg}
    \includegraphics{figs/defl-poly-ddual-vp.mps}}
  \caption{\protect\subref{fig:vis-path:recur} A node $x_i$ of a
    directed dual and its neighbours $x_{i-1}$, $r$ and $\ell$ in an
    iteration of the construction of a visibility
    path, \protect\subref{fig:vis-path:indseq} a deflated polygon
    triangulation, $T$, wherein the induced sequence of the vertex $u$
    through the edge $e$ is $(e$, $f$, $g$, $h)$
    and \protect\subref{fig:vis-path:eg} the directed dual, $T$, in
    which the visibility path of the directed dual starting with nodes
    $(a,b)$ is $(a$, $b$, $c$, $d$, $h^T)$.}
  \label{fig:vis-path}
\end{figure}

Let $u$ be the vertex of a deflated polygon triangulation, $T$, and
let $e$ be an edge opposite $u$ in a triangle of $T$.
An \emph{induced sequence} of $u$ through $e$ is the sequence of edges
through which $u$ sees a polygon edge, $f$, through $e$.  This
sequence is ordered by the proximity to $u$ of their intersections
with a closed line segment joining $u$ and $f$ that is interior to the
open polygon everywhere but at its endpoints (e.g. see
Figure~\ref{fig:vis-path:indseq}).

\begin{lemma}
  \label{lem:defl-vis-thr}
  Suppose $u$ is a vertex opposite a closed non-polygon edge, $e$, in
  a triangle, $a$, of a deflated polygon triangulation.  Let $v$ be
  the reflex endpoint of $e$ and let $f$ be the edge opposite $v$ in
  the triangle sharing $e$ with $a$ (see Figure~\ref{fig:rrefl}).
  Then $u$ sees the same polygon edge through $e$ as $v$ sees through
  $f$.
\end{lemma}
\iffullversion
\begin{proof}
  The ray from $u$ through $v$ intersects $f$.  Then, if $f$ is a
  polygon edge, $u$ sees it.  Otherwise, $f$ is a diagonal and the ray
  intersects some other polygon edge visible to both $u$ and $v$.
  From Lemma~\ref{lem:defl-vis-thr-uniq}, we have the uniqueness of
  the edge $u$ sees through $f$, which completes the proof.
\end{proof}
\fi

\begin{cor}
  \label{cor:ind-path-rec}
  If $u$, $v$, $e$ and $f$ are as in Lemma~\ref{lem:defl-vis-thr},
  then the induced sequence of $u$ through $e$ is equal to that of $v$
  through $f$ prepended with $e$.
\end{cor}
\iffullversion
\begin{proof}
  This follows from Lemma~\ref{lem:defl-vis-thr} and
  Corollary~\ref{cor:subpoly-defl}.
\end{proof}
\fi


\subsection{Directed Duals and Visibility}

A \emph{visibility path}, $(x_1$, $x_2$, \ldots, $x_n)$, of
the \emph{directed dual}, $D$, of a deflated polygon is a sequence of
nodes in $D$ meeting the following conditions.  $x_1$ is a triangle
node adjacent to $x_2$ and, for $i\in\{2,\ldots,n\}$, if $x_i$ is a
terminal node, then it is $x_n$---the final node of the path.
Otherwise, let the neighbours of $x_i$ be $x_{i-1}$, $r$ and $\ell$ in
counter-clockwise order (see Figure~\ref{fig:vis-path:recur}).  Then
\[ x_{i+1} = 
\begin{cases}
  r & \text{if edge $\{x_{i-1}, x_i\}$ is directed $x_{i-1}\leftarrow x_i$} \\
  \ell & \text{if edge $\{x_{i-1}, x_i\}$ is directed $x_{i-1}\rightarrow x_i$} \\
\end{cases}
\]
(e.g. see Figure~\ref{fig:vis-path:eg}).

\iffullversion
Note that two consecutive nodes of a visibility path determine all
subsequent nodes and so any suffix of length greater than one of a
visibility path is also a visibility path.
\fi

\begin{lemma}
  \label{lem:vis-path-geom}
  Let $(a,b,c)$ be a simple path in the directed dual, $D$, of a
  deflated polygon triangulation, $T$, where $a$ and $b$ are triangle
  nodes joined by the edge $e$.  Let $u$ be the vertex opposite $e^T$
  in $a^T$, let $v$ be the reflex endpoint of $e^T$ and let $f$ be the
  edge opposite $v$ in $b^T$ (see Figure~\ref{fig:vis-path:indseq}).
  Then $(a,b,c)$ is the substring of a visibility path if and only if
  $f^D$ joins $b$ and $c$ in $D$.
\end{lemma}
\begin{proof}
  Suppose $(a,b,c)$ is the substring of a visibility path and let $x$
  be the neighbour of $b$ not $a$ nor $c$ and let $x'$ be the edge of
  $b^T$ not $e^T$ nor $f$. We consider the case where the neighbours
  of $b$ are $a$, $x$ and $c$ in counter-clockwise order---the
  argument is symmetric in the other case.  Then $(a,b)$ is
  right-reflex and so $b^T$ has counter-clockwise edge ordering:
  $e^T$, $x'$, $f$.  Then, since edge orderings are preserved in the
  directed dual, $f^D$ joins $b$ and $c$ as required.  Reversing the
  argument gives the converse.
\end{proof}

\begin{cor}
  \label{cor:ind-path-vis}
  Let $D$, $T$, $a$, $b$, $e$ and $u$ be as in
  Lemma~\ref{lem:vis-path-geom}.  The induced sequence of $u$ through
  $e$ is equal to the sequence of correspondents in $T$ of edges
  traversed by the visibility path starting with $(a,b)$ in $D$.  The
  final node of this visibility path corresponds to the edge $u$ sees
  through $e^T$.
\end{cor}
\iffullversion
\begin{proof}
  This follows, by induction, from Lemma~\ref{lem:vis-path-geom} and
  Corollary~\ref{cor:ind-path-rec}.
\end{proof}
\fi

\begin{theorem}
  \label{thm:dd-det-ve}
  A vertex, $u$, and edge, $g$, of a deflated polygon, $P$, are
  visible if and only if there is a visibility path in the directed
  dual, $D$, of the triangulation, $T$, of $P$ starting on a triangle
  node corresponding to a triangle incident to $u$ and ending on
  $g^D$.
\end{theorem}
\begin{proof}
  Assume $u$ sees $g$.  If $g$ is an edge of a triangle, $a$, incident
  to $u$ then $(a^D, g^D)$ is the required visibility path.  Otherwise
  $u$ sees $g$ through some edge, $e$, and the existence of the
  required visibility path follows from
  Corollary~\ref{cor:ind-path-vis}.
  
  Assume, now, that the visibility path exists.  If its triangle nodes
  all correspond to triangles incident to $u$ then $g$ is incident to
  one of these triangles and so visible to $u$.  Otherwise, let $e$ be
  the first edge the path traverses from a node, $a$, corresponding to
  a triangle incident to $u$ to a node, $b$, corresponding to a
  triangle not incident to $u$.
  
  Then, by Corollary~\ref{cor:ind-path-vis}, the induced sequence of
  $u$ through $e^T$ corresponds to a visibility path starting with
  $(a, b)$ and this visibility path ends on a node corresponding to
  the edge $u$ sees through $e$.  Since two consecutive nodes of a
  visibility path determine all subsequent nodes, these visibility
  paths end on the same node, $g^D$, and so $u$ sees $g$.
\end{proof}


\begin{figure}[htb]
  \centering \includegraphics{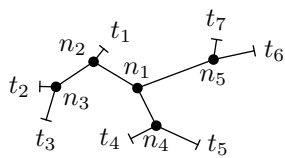}
  \caption{A plane tree with the following maximal outer paths:
    $(t_7$, $n_5$, $n_1$, $n_2$, $t_1)$, $(t_1$, $n_2$, $n_3$, $t_2)$,
    $(t_2$, $n_3$, $t_3)$, $(t_3$, $n_3$, $n_2$, $n_1$, $n_4$, $t_4)$,
    $(t_4$, $n_4$, $t_5)$, $(t_5$, $n_4$, $n_1$, $n_5$, $t_6)$,
    $(t_6$, $n_5$, $t_7)$.}
  \label{fig:outer-path}
\end{figure}

An \emph{outer path} of a plane tree, $D$, is the sequence of nodes
visited in a counter-clockwise walk along its outer face in which no
node is visited twice.  An outer path is \emph{maximal} if it is not a
proper substring of any other outer path (e.g. see
Figure~\ref{fig:outer-path}).  Note that an outer path, $(x_1$, $x_2$,
\ldots, $x_n)$, of the directed dual of a polygon triangulation, $T$,
corresponds to a triangle fan in $T$ where the triangles have
clockwise order $x_1^T, x_2^T, \ldots, x_n^T$ about their shared
vertex.

\begin{theorem}
  A pair of vertices, $u$ and $v$, of a deflated polygon $P$ are
  visible if and only if, in the directed dual, $D$, of the
  triangulation, $T$, of $P$, their corresponding maximal outer paths
  share a node.
\end{theorem}
\begin{proof}
  The maximal outer paths of $u$ and $v$ share a node in $D$ if and
  only if they are incident to a common triangle in $T$ and, since $P$
  is deflated, this is the case if and only if $u$ and $v$ are
  visible.
\end{proof}


\subsection{Directed Dual Equivalence}

In this section, we show that if two deflated polygons have the same
directed dual, then one can be monotonically deformed into the other.
First, we fully characterize the directed duals of deflated polygons.

\begin{figure}[htb]
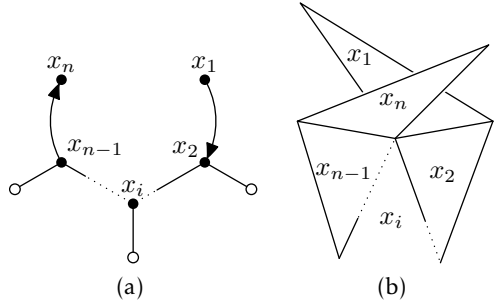

  \centering \subfloat[]{\label{fig:ill-path:tree}
    \includegraphics{figs/ill-path.mps}} \quad
  \subfloat[]{\label{fig:ill-path:olap}
    \includegraphics{figs/ill-path-poly.mps}}
  \caption{If \protect\subref{fig:ill-path:tree} the tree with outer
    path $(x_1$, $x_2$, \ldots, $x_n)$ were a subtree of the directed
    dual of a polygon triangulation, $T$,
    then \protect\subref{fig:ill-path:olap} the triangles
    corresponding to nodes $x_1$, $x_2$, $x_{n-1}$ and $x_n$ in $T$
    would overlap, contradicting the simplicity of the
    polygon.}  \label{fig:ill-path}
\end{figure}

\begin{figure}[htb]
  \centering
  \includegraphics{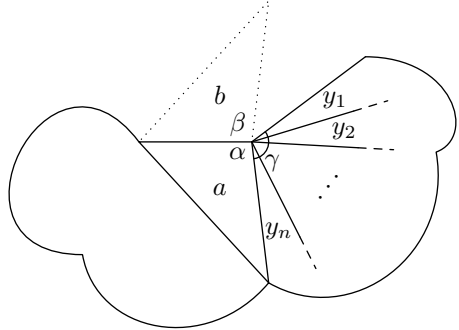}
  \caption{The inductive polygon in the proof of
    Theorem~\ref{thm:ddd-iff-nip} or a polygon from the inductive
    deformation in the proof of Theorem~\ref{thm:defl-dd-eq}.}
  \label{fig:no-ill-then-space}
\end{figure}

\begin{theorem}
  \label{thm:ddd-iff-nip}
  A partially directed plane tree, $D$, in which every non-terminal
  node has degree three and where an edge is directed if and only if
  it joins two non-terminal nodes of degree three is the directed dual
  of a deflated polygon if and only if it does not contain an outer
  path, $(x_1$, $x_2$, \ldots, $x_n)$, with $n\ge 4$, such that the
  edges from $x_1$ and $x_{n-1}$ are both forward directed (i.e.
  $x_1\rightarrow x_2$ and $x_{n-1}\rightarrow x_n$).
\end{theorem}
Henceforth, we call such a path an \emph{illegal path}.
\begin{proof}
  Suppose $D$ contains an illegal path, $(x_1$, $x_2$, \ldots, $x_n)$.
  If $D$ is the directed dual of a polygon triangulation, $T$, then
  $x_1^T$, $x_2^T$, $x_{n-1}^T$ and $x_n^T$ share a common vertex
  reflex in both quadrilaterals $x_1^T\cup x_2^T$ and $x_{n-1}^T\cup
  x_n^T$ (see Figure~\ref{fig:ill-path}).  This contradicts the
  disjointness of these quadrilaterals.
  
  Suppose, now, that $D$ has no illegal paths.  We prove the converse
  with a construction of a polygon triangulation having $D$ as its
  directed dual.  Let $b$ be a terminal node in the subtree of $D$
  induced by its non-terminal nodes.  Then $b$ has two terminal
  neighbours and one non-terminal neighbour, $a$.  Let $D'$ be the
  tree resulting from replacing $a$ and its terminal neighbours with a
  single terminal node, $x$, connected to $b$ with an undirected edge.
  By induction on the number of non-terminal nodes, there exists a
  deflated polygon triangulation, $T$, having $D'$ as its directed
  dual.
  
  Assume, without loss of generality, that the edge joining $a$ and
  $b$ is directed $a\rightarrow b$.  Let $u$ be the endpoint of
  $x^{T}$ pointing in a clockwise direction in the boundary of $T$ and
  let $(y_1$, $y_2$, \ldots, $y_n)$ be the outer path of $D$
  corresponding to the triangles other than $b^{T}$ in $T$ incident to
  $u$ (see Figure~\ref{fig:no-ill-then-space}).  Note that $(y_i$,
  $y_{i+1}$, \ldots, $y_n$, $a$, $b)$ is an outer path of $D$ and so,
  by hypothesis, for all $i\in\{1, 2, \ldots, n-1\}$, the edge joining
  $y_i$ and $y_{i+1}$ is directed $y_i\leftarrow y_{i+1}$.
  
  Then, to show that a triangle may be appended to $T$ to form the
  required triangulation, it suffices to show that the sum of the
  angles at $u$ of the triangles $y_1^{T}$, $y_2^{T}$, \ldots,
  $y_n^{T}$ is less than $\pi$, which, in turn, follows from the
  backward directedness of the edges of $(y_1$, $y_2$, \ldots, $y_n)$.
\end{proof}


\begin{theorem}
  \label{thm:defl-dd-eq}
  If the deflated polygons $P$ and $P'$ have the same directed dual,
  $D$, then $P$ can be monotonically deformed into $P'$.
\end{theorem}
\begin{proof}
  Let $b$ be an ear of the triangulation, $T$, of $P$ and let $b'$ be
  the triangle corresponding to $b^D$ in the triangulation, $T'$, of
  $P'$.  By induction on the number of triangles in $T$, there is a
  monotonic deformation $t\mapsto Q^t$ from $Q = P\setminus b$ to $Q'
  = P'\setminus b'$.  Note that replacing $b^D$ and its terminal nodes
  in $D$ with a single terminal node gives the directed dual, $D'$, of
  $Q$.  Then, since $Q$ is deflated (Lemma~\ref{lem:cutear-defl}) and
  $t\mapsto Q^t$ is monotonic, for all $t\in [0,1]$, $Q^t$ is deflated
  and has directed dual $D'$.
  
  Let $v$ be the helix of $b$, let $a$ be the triangle sharing an
  edge, $e$, with $b$ and let $u$ be the reflex endpoint of $e$.  We
  need to show that there is a continuous map $t\mapsto v^t$ that,
  combined with $t\mapsto Q$, gives a monotonic deformation of a
  polygon with directed dual $D$.  For $t\in [0,1]$, let $\alpha^t$ be
  the angle of $a^t$ at $u^t$ in $Q^t$ and let $\gamma^t$ be the sum
  of the angles at $u^t$ of the triangles, $y_1^t$, $y_2^t$, \ldots,
  $y_n^t$, other than $a^t$ of the triangulation of $Q^t$ incident to
  $u^t$ (see Figure~\ref{fig:no-ill-then-space}).
  
  Then, since $v$ may be brought arbitrarily close to $u$ in a
  monotonic deformation of $P$, it suffices to show that there is a
  continuous map $t\mapsto \beta^t$ specifying an angle for $b^t$ at
  $u^t$ such that, for all $t\in [0,1]$, $0 < \beta^t < \pi$,
  $\alpha^t + \beta^t > \pi$ and $\alpha^t + \beta^t + \gamma^t <
  2\pi$.  The latter two conditions are equivalent to
  \[
    \pi - \alpha^t < \beta^t < (\pi - \alpha^t) + (\pi - \gamma^t)\;.
  \]
  It follows from Theorem~\ref{thm:ddd-iff-nip} that the outer path
  $(y_1^{D'}$, $y_2^{D'}$, \ldots, $y_n^{D'})$ is left-directed and so
  that $\gamma^t < \pi$.  Then $\beta^t = \pi - (\alpha^t
  + \gamma^t)/2$ satisfies all required conditions.
  
  Now, let $t\mapsto R^t$ be the monotonic deformation from a polygon
  with directed dual $D$ combining $t\mapsto Q^t$ and the map
  $t\mapsto v^t$ defined by a fixed distance between $u^t$ and $v^t$
  of $r\in\R_{>0}$ and an angle for $b^t$ at $u^t$ of
  $\beta^t$.
  
  Prepending $t\mapsto R^t$ with a deformation of $P$ in which $v$ is
  brought to the distance $r$ from $u$ and then rotated about $u$ to
  an angle of $\beta^0$; then appending a deformation comprising
  similar motions ending at $P'$; and, finally, scaling in time gives
  a continuous map, $t\mapsto P^t$, with $P^0 = P$ and $P^1 = P'$.
  Since, for all $t\in[0,1]$, $Q^t$ is simple, a small enough $r$ can
  be chosen such that $t\mapsto P^t$ is simplicity-preserving.  Then,
  by the properties of $t\mapsto \beta^t$, $t\mapsto P^t$ is the
  required monotonic deformation.
\end{proof}


\iffullversion

\section{Vertex-Edge Visibilities in Monotonic Deformations}

In the following Lemmas we use analytic arguments similar to those
used by {\'A}brego et al. \cite{Abrego11} to investigate the nature of
collinearities in deformations and derive a needed vertex-edge
visibility property of monotonic deformations.

\begin{lemma}
  Let $t\mapsto P^t$ be a deformation of a polygon, $P$, let $u$, $v$
  and $w$ be vertices of $P$ and let $c\in[0,1]$.  Suppose that, for
  every $\delta>0$, the pierced $\delta$-neighbourhood, $N_\delta =
  (c-\delta, c+\delta)\cap [0,1]\setminus \{c\}$, of $c$ has a point,
  $s\in N_\delta$, such that $u^s$, $v^s$ and $w^s$ are collinear in
  $P^s$.  Then $u^c$, $v^c$ and $w^c$ are collinear in $P^c$.
\end{lemma}
\begin{proof}
  Assume otherwise and, for all $t\in[0,1]$, let $\alpha^t$ be the
  angle between $u^t$, $v^t$ and $w^t$ in $P^t$.  Then $t\mapsto
  \alpha^t$ is continuous and $\alpha^c\neq k\pi$, for any $k\in\Z$.
  Then, by hypothesis, there is no $\delta>0$ such that, for all $t\in
  N_\delta$, $|\alpha^c - \alpha^t| < \min_{k\in\Z}|\alpha^c-k\pi| >
  0$, contradicting the continuity of $t\mapsto \alpha^t$.
\end{proof}

\begin{cor}
  Let $t\mapsto P^t$ be a deformation of a polygon, $P$, let $u$, $v$
  and $w$ be vertices of $P$ and let $c\in[0,1]$.  If $u^c$, $v^c$ and
  $w^c$ are not collinear in $P^c$, then there exists a $\delta>0$
  such that, for every $t\in N_\delta$, $u^t$, $v^t$ and $w^t$ are not
  collinear in $P^t$.
\end{cor}

\begin{cor}
  \label{cor:safe-neighb}
  Let $t\mapsto P^t$ be a deformation of a polygon, $P$, and let
  $c\in[0,1]$.  There exists a $\delta>0$ such that, for every $t\in
  N_\delta$, no three vertices are collinear in $P^t$ unless their
  correspondents are collinear in $P^c$.
\end{cor}

We call the corresponding $\delta$-neighbourhood, $N_\delta$, the
\emph{safe neighbourhood} of $c$.

\begin{lemma}
  \label{lem:rad-pres-neighb}
  Let $t\mapsto P^t$ be a deformation of a polygon, $P$, let $u$ be a
  vertex of $P$, let $c\in[0,1]$, let $N_\delta$ be a safe
  neighbourhood of $c$ and let $W^c$ be a subset of the vertices of
  $P^c$ having distinct projections onto the unit circle about $u^c$.
  Then, for all $t\in N_\delta$, the corresponding vertex subset,
  $W^t$, of $P^t$ has the same radial order about $u^t$ as does $W^c$
  about $u^c$.
\end{lemma}
\begin{proof}
  Since deformations preserve simplicity, the vertices of $P^t$ never
  coincide and so their projections on the unit circle about $u^t$
  also move continuously.  Then a change in radial order between two
  vertices, say $v$ and $w$, implies that, for some intermediate
  $c'\in(c,t)$, $u^{c'}$, $v^{c'}$ are $w^{c'}$ collinear in $P^{c'}$,
  contradicting $N_\delta$ being a safe neighbourhood.
\end{proof}


\begin{figure}[htb]
  \centering
  \includegraphics{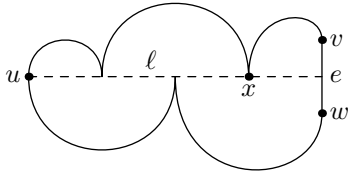}
  \caption{A polygon with visible vertex-edge pair, $(u,e)$, joined by
    a unique closed line segment, $\ell$, contained in the closed
    polygon.}
  \label{fig:u-see-edge}
\end{figure}

\begin{figure}[htb]
  \centering
  \includegraphics{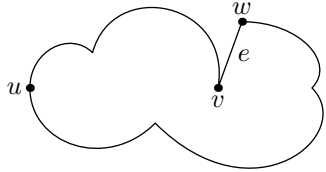}
  \caption{A polygon in general position in which a vertex, $u$, and
    edge, $e$, are not visible and where $u$ sees a single endpoint,
    $v$, of $e$.  In such a polygon, $e$ necessarily neither faces nor
    is collinear to $u$.}
  \label{fig:edge-mon-see-ep}
\end{figure}

\else 

\section{Deflatability of Polygons}

In this section, we show how deflatable polygons may be related
combinatorially to their deflation targets and use this result to
present a polygon that cannot be deflated.  We also show that
vertex-vertex visibilities do not determine deflatability.  These
results depend on the following Lemma.

\fi

\begin{lemma}
  \label{lem:mon-defm-ve-mon}
  Let $t\mapsto P^t$ be a monotonic deformation of a polygon, $P$, in
  general position.  Then a vertex and an edge are visible in $P^1$
  only if they are visible in $P$.
\end{lemma}

\iffullversion

\begin{proof}
  Suppose a vertex, $u^1$, sees an edge, $e^1$, with endpoints $v^1$
  and $w^1$ in $P^1$ such that $u$ and $e$ are not visible in $P$.
  Let $c$ be the supremum of the set \[\{x\in[0,1] : \text{for all
  $t\in[0,x)$, $u^t$ and $e^t$ are not visible in $P^t$}\}\] and let
  $N_\delta$ be a safe neighbourhood of $c$
  (Corollary~\ref{cor:safe-neighb}).  Note that $e^c$ is either facing
  $u^c$ or is collinear with $u^c$ in $P^c$, since, otherwise, for all
  $t\in N_\delta$, $u^t$ and $e^t$ would not be visible in $P^t$,
  contradicting the choice of $c$.
  
  We begin by establishing the claim that there exists a unique closed
  line segment, $\ell$, contained in the closed polygon $P^c$ joining
  $u^c$ and the closed edge $e^c$.  Suppose, first, that no such
  segment exists.  Then every open line segment joining $u^c$ and
  $e^c$ intersects an open edge of $P^c$.  Then, by
  Lemma~\ref{lem:rad-pres-neighb}, for all $t\in N_\delta$, $u^t$ and
  $e^t$ are not visible in $P^t$, contradicting the choice of $c$.
  
  Suppose, now, that two such segments exist.  Then either these
  segments are collinear, in which case so are $u^c$ and the endpoints
  of $e^c$, contradicting the monotonicity of the deformation (since
  $u$ may see at most a single endpoint of $e$ in $P$ without seeing
  $e$ itself) or else they form a triangle.  Then every closed segment
  joining $u^c$ and $e^c$ contained in this closed triangle is also
  such a segment and so, by Lemma~\ref{lem:rad-pres-neighb}, for all
  $t\in N_\delta$, $u^t$ and $e^t$ are visible in $P^t$, contradicting
  the choice of $c$.
  
  From the claim, it follows that $e^c$ is facing $u^c$ in $P^c$ and
  we are left with two cases.

  \noindent\textbf{Case I:} $\ell$ joins $u^c$ and a point on
  the open edge $e^c$.  Since $\ell$ is unique, there is at least one
  vertex from each of the two chains of $P^c$ from $u^c$ to $e^c$
  incident to $\ell$, as in Figure~\ref{fig:u-see-edge}.  Let $x^c$ be
  the furthest of these vertices from $u^c$ and let $s\in N_\delta$,
  with $s<c$.  Suppose $u^s$ and $x^s$ are visible in $P^s$.  Then the
  closed line segment joining $u^s$ and $x^s$ is contained in the
  closed polygon $P^s$ but, since $s<c$, the extension of this segment
  joining $x^s$ and $e^s$ must intersect an open edge $f^s$ of $P^s$.
  It, then, follows from Lemma~\ref{lem:rad-pres-neighb} that an
  endpoint of $f^c$ is incident to $\ell$ in $P^c$ between $x^c$ and
  $e^c$, contradicting the choice of $x^c$.  Then $u^s$ and $x^s$ are
  not visible in $P^s$ but $u^c$ and $x^c$ are visible in $P^c$,
  contradicting the monotonicity of the deformation.

  \noindent\textbf{Case II:} $\ell$ joins $u^c$ and an
  endpoint, say $v^c$ without loss of generality, of $e^c$.  Then
  $u^c$ sees $v^c$ in $P^c$ and so, by monotonicity, for all
  $t\in[0,c)$, $u^t$ sees $v^t$ in $P^t$.  Since $P$ is in general
  position and $u$ and $e$ are not visible in $P$, $e$ must neither be
  facing $u$ nor be collinear with $u$ in $P$, as in
  Figure~\ref{fig:edge-mon-see-ep}.  Then, since $e^c$ is facing $u^c$
  in $P^c$, there must be some intermediate $c'\in(0,c)$ such that
  $u^{c'}$ and $e^{c'}$ are collinear in $P^{c'}$.  But since $u^{c'}$
  sees $v^{c'}$ in $P^{c'}$, it must also see the other endpoint,
  $w^{c'}$, of $e^{c'}$, contradicting the monotonicity of the
  deformation.
\end{proof}

\section{Deflatability of Polygons}

With this result, we now show how deflatable polygons may be related
combinatorially to their deflation targets and use this result to
present a polygon that cannot be deflated.  We also show that
vertex-vertex visibilities do not determine deflatability.

\else 

The proof, which is available in the full version of this
paper \cite{Bose12}, uses analytic arguments similar to those used by
{\'A}brego et al. \cite{Abrego11}.

\fi


A \emph{compatible directed dual} of a polygon, $P$, in general
position is the directed dual of a deflated polygon, $P'$, such that,
under an order- and chirality-preserving bijection between the
vertices of $P$ and $P'$, a vertex-edge or vertex-vertex pair are
visible in $P'$ only if their correspondents are visible in $P$.  By
\emph{chirality-}preserving bijection, we mean one under which a
counter-clockwise walk on the boundary of $P$ corresponds to a
counter-clockwise walk on the boundary of $P'$.

\begin{theorem}
  \label{thm:no-cdd-not-defl}
  A polygon, $P$, in general position with no compatible directed dual
  is not deflatable.
\end{theorem}
\begin{proof}
  It follows from Lemma~\ref{lem:mon-defm-ve-mon} that if $P$ is
  monotonically deformable to a deflated polygon $P'$, then the
  directed dual of $P'$ is compatible with $P$.
\end{proof}

\begin{lemma}
  \label{lem:cdd-from-dd}
  Suppose a polygon, $P$, in general position has a compatible
  directed dual, $D$.  Let $P'$ be the deflated polygon with directed
  dual $D$ whose vertex-vertex and vertex-edge visibilities are a
  subset of those of $P$ under an order- and chirality-preserving
  bijection.  Then the unique triangulation, $T'$, of $P'$ is a
  triangulation, $T$, of $P$ under the bijection and $D$ can be
  constructed by directing the undirected non-terminal edges of the
  directed dual of $T$.
\end{lemma}
\begin{proof}
  Note that $T'$ is the visibility graph of $P'$.  Then, since $P$ is
  in general position and has the same vertex count as $P'$, it
  follows from the vertex-vertex visibility subset property of $P'$
  that $T'$ triangulates $P$ under the bijection.
  
  It remains to show that, for every non-terminal edge of the directed
  dual of $T$, either the edge is undirected or it is directed as in
  $D$ or, equivalently, that for every pair of adjacent triangles, $a$
  and $b$, in $T$ corresponding to the triangles $a'$ and $b'$ in
  $T'$, if $(a,b)$ is right-reflex then so is $(a',b')$.  Suppose,
  instead, that $(b', a')$ is right-reflex.  Let $e'$ be the edge
  shared by $a'$ and $b'$, let $u'$ be the vertex of $a'$ opposite
  $e'$ and let $f'$ be the edge of $b'$ opposite the reflex endpoint
  of $e'$.  Then, by Lemma~\ref{lem:defl-vis-thr}, $u'$ sees an edge
  through $f'$ but the corresponding visibility is not present in $P$,
  contradicting the vertex-edge visibility subset property of $P'$.
\end{proof}


\begin{figure}[htb]
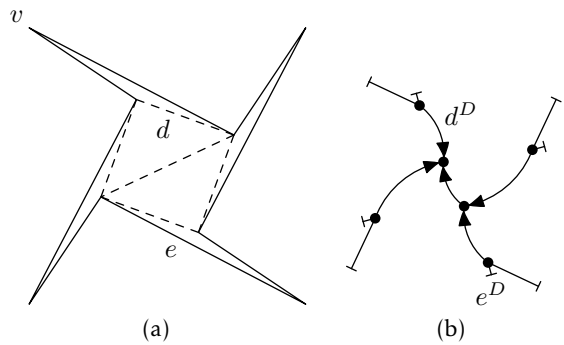

  \centering
  \subfloat[]{\label{fig:counter-poly:poly}
    \includegraphics{figs/counter-poly.mps}}
  \quad
  \subfloat[]{\label{fig:counter-poly:ddcand}
    \includegraphics{figs/counter-poly-ddual-cand.mps}}
  \caption{\protect\subref{fig:counter-poly:poly} A non-deflatable
    polygon, $P$, with its only triangulation, up to symmetry,
    indicated with dashed lines
    and \protect\subref{fig:counter-poly:ddcand} its only candidate
    for a compatible directed dual, $D$, up to symmetry. }
  \label{fig:counter-poly}
\end{figure}

\begin{theorem}
  There exists a polygon that cannot be deflated.
\end{theorem}
\begin{proof}
  We show that the general position polygon, $P$, in
  Figure~\ref{fig:counter-poly:poly} has no compatible directed dual
  and so, by Lemma~\ref{thm:no-cdd-not-defl}, is not deflatable.
  Assume that the directed dual, $D$, of a deflated polygon, $P'$, is
  compatible with $P$.  Then, by Lemma~\ref{lem:cdd-from-dd}, $D$ can
  be constructed by directing the undirected non-terminal edges of the
  directed dual of some triangulation of $P$.  Up to symmetry, $P$ has
  a single triangulation, its directed dual has a single undirected
  non-terminal edge and there is a single way to direct this edge.
  Then we may assume, without loss of generality, that $D$ is the tree
  shown in Figure~\ref{fig:counter-poly:ddcand} and, by
  Theorem~\ref{thm:dd-det-ve}, the correspondents of the vertex $v$
  and edge $e$ in $P'$ are visible.  This contradicts the
  compatibility of $D$.
\end{proof}


\iffullversion

\begin{figure}[htb]
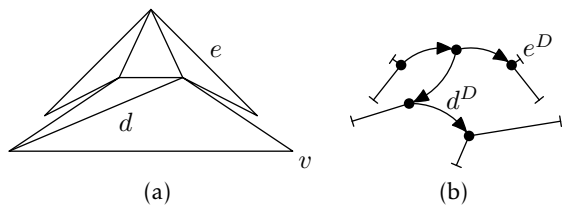

  \centering
  \subfloat[]{\label{fig:ndef-hept:tri}
    \includegraphics{figs/ndef-heptagon-tri.mps}}
  \quad
  \subfloat[]{\label{fig:ndef-hept:ddcand}
    \includegraphics{figs/ndef-heptagon-ddual-cand.mps}}
  \caption{\protect\subref{fig:ndef-hept:tri} The only triangulation,
    up to symmetry, of a non-deflatable heptagon
    and \protect\subref{fig:ndef-hept:ddcand} its only candidate for a
    compatible directed dual, $D$.}
  \label{fig:ndef-hept}
\end{figure}

\begin{figure}[htb]
  \centering
  \includegraphics{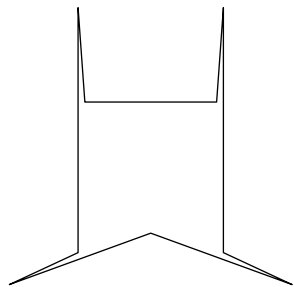}
  \caption{A non-deflatable nonagon.}
  \label{fig:ndef-non}
\end{figure}

This combinatorial technique of using the non-existence of compatible
directed duals can also be applied to other polygons.  For example,
the heptagon whose unique triangulation (up to symmetry) is shown in
Figure~\ref{fig:ndef-hept:tri} can also be shown to be non-deflatable
in this way, as can the nonagon of Figure~\ref{fig:ndef-non}.

Specifically, the directed dual of the triangulation of the heptagon,
$P$, leaves only a single non-terminal edge undirected.  By
Theorem~\ref{thm:ddd-iff-nip}, this edge may be directed in only one
way such that the resulting directed dual, $D$, shown in
Figure~\ref{fig:ndef-hept:ddcand}, has no illegal paths and is thus
the directed dual of a deflated polygon, $P'$.  Then the
correspondents of $v$ and $e$ are visible in $P'$ so that $D$ is
incompatible with $P$.  The non-deflatability proof for the nonagon
must consider two triangulations but has the same general form.

\else 

Note that, although the non-deflatability of $P$ can be shown using
\emph{ad hoc} arguments, the combinatorial technique used here can be
applied to other polygons.  See the full version of this
paper \cite{Bose12} for examples.

\fi

\iffullversion

A natural question is: What is the least $n$ for which there exists a
non-deflatable $n$-gon in general position?  It is trivial to show
that every quadrilateral is deflatable and not difficult to show the
same for all general position pentagons.  Then it remains only to
check for the existence of a non-deflatable hexagon.

\begin{theorem}\label{thm:hex-deflate}
  Every general position hexagon is montonically deflatable.
\end{theorem}

\begin{proof}
  Let $P$ be any hexagon in general position.  The proof is by
  induction on the number, $m$, of pairs of mutually visible
  non-adjacent vertices of $P$.  The base case, $m=3$, happens when
  the hexagon is already deflated (it has a unique triangulation with
  four triangles and three non-polygon edges).

  The inductive step is made using an enormous case analysis grouped
  by the number of reflex vertices of the hexagon.  Note that the
  vertex set of a general position hexagon (segments joining visible
  vertex pairs are interior to the hexagon) may be put into point set
  general position (no three points are collinear) through a
  visibility-preserving perturbation.  Thus we may assume in this step
  that the vertex set of $P$ is in general position.

  The case where $P$ has no reflex vertices is handled by moving a
  vertex inwards until it becomes reflex.  A simple argument,
  presented in the next paragraph, suffices to handle all cases in
  which $P$ has exactly one reflex vertex.

  Refer to Figure~\ref{fig:hexdef-one-reflex} for what follows.  Let
  $a$, $b$, $c$, $d$, $e$, and $f$ be the vertices of $P$ in the order
  they occur on the boundary of $P$ and suppose, without loss of
  generality, that $a$ is the unique reflex vertex of $P$.  Suppose,
  again without loss of generality that there is a closed halfplane
  with $a$ on its boundary that contains $a$, $b$, $c$, and $d$, but
  not $f$. Then $abcd$ is a convex quadrilateral contained in $P$ and
  moving $c$ directly towards $a$ until it crosses $bd$ removes at
  least one visible pair, namely $bd$, from $P$.  This motion is
  monotonic because the only vertices not visible from $c$ (possibly
  $f$ and $e$) remain hidden ``behind'' $a$.  In particular, the
  orientiations of the triangles $fac$ and $eac$ do not change during
  this motion.

  \begin{figure}[htb]
    \centering
    \includegraphics{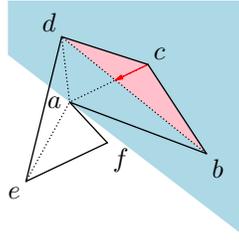}
    \caption{The inductive step of Theorem~\ref{thm:hex-deflate} when $P$ has
      one reflex vertex.}
    \label{fig:hexdef-one-reflex}
  \end{figure}

  The remaining cases have 2 or 3 reflex vertices and are each handled
  using a motion illustrated in Appendix~\ref{apx:hexagons}.  All
  these motions move a single vertex, say $a$, along a linear
  trajectory until it crosses a segment joining a visible pair of
  vertices in $P$.  All these motions have two properties that make it
  easy to check their correctness:

  \begin{enumerate}
  \item There is a convex polygon, $C$, whose vertices are a subset of
    those of $P$, that contains $a$, $b$, and $f$.  Additionally, the
    closure of $C$ contains $ab$ and $af$, while its interior
    intersects the boundary of $P$ at most in $ab$ and $af$.
    Throughout the motion, $a$ remains within $C$, except at the end,
    where it passes through the interior of an edge of $C$ that is
    interior to $P$ and stops an arbitrarily small distance outside of
    $C$. This guarantees that $P$ remains simple throughout the
    motion.  (See Figure~\ref{fig:hexdef-example}, where $C$ is the
    triangle $bcf$.)
  
  \item
    The motion of $a$ is such that the region bounded by the polygon
    strictly loses points, i.e., $P^{t'} \subseteq P^t$ for all $0\le
    t\le t'\le 1$.  This ensures that no pair of vertices
    $x,y\in\{b,c,d,e,f\}$ ever becomes visible during the motion.
    That is, the only possibility of the motion being non-monotonic
    comes from the possibility that $a$ may gain visibilities as it
    moves.
  \end{enumerate}

  The only remaining check, for each case, is to ensure that no new
  visible pair involving $a$ appears during the motion. This can be
  done case by case using only order type information about $P$.  We
  now illustrate one example, in Figure~\ref{fig:hexdef-example}.  In
  this example, $a$ is moved toward the interior of $P$ along the line
  through $ab$ until it crosses the segment $fc$.  This eliminates the
  visible pair $fc$.  This motion satisfies properties 1 and 2, above,
  so the polygon remains simple throughout the motion and no new
  visible pairs not involving $a$ are created.  To check that no new
  visible pair involving $a$ is created during the motion, observe
  that, initially, the only vertex not visible from $a$ is $e$.  In
  particular, this is because the sequence $efa$ forms a right turn.
  This remains true at the end of the motion because $efc$ forms a
  right turn and, at the end of the motion, $a$ is arbitrarily close
  to the segment $fc$.  Therefore, by convexity, $efa$ forms a right
  turn throughout the motion and at no point during the motion does
  the pair $ae$ become visible.

  \begin{figure}[htb]
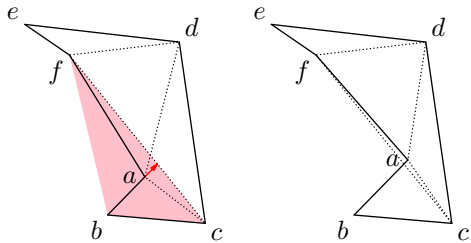

    \centering
    \subfloat{\includegraphics{figs/hexdef-example-before.mps}}
    \quad
    \subfloat{\includegraphics{figs/hexdef-example-after.mps}}
    \caption{One of the cases where $P$ has more than one reflex
      vertex in the proof of Theorem~\ref{thm:hex-deflate}.}
    \label{fig:hexdef-example}
  \end{figure}

  Similar statements can be verified for all the polygons in
  Appendix~\ref{apx:hexagons}.  We wish the reader good luck with
  their verification.
\end{proof}

\fi


\begin{figure}[htb]
  \centering
  \includegraphics{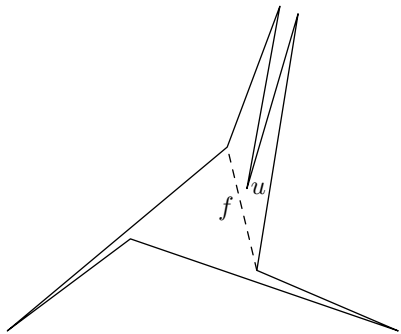}
  \caption{A deflatable polygon with the same vertex-vertex
    visibilities as the non-deflatable polygon shown in
    Figure~\ref{fig:counter-poly:poly}.}
  \label{fig:def-poly-same-vis}
\end{figure}

\begin{theorem}
  The vertex-vertex visibilities of a polygon do not determine its
  deflatability.
\end{theorem}
\begin{proof}
  The polygon in Figure~\ref{fig:def-poly-same-vis} has the same
  vertex-vertex visibilities as the non-deflatable polygon in
  Figure~\ref{fig:counter-poly:poly} and yet can be deflated by moving
  the vertex $u$ through the diagonal $f$.
\end{proof}


\section{Summary and Conclusion}

We presented the directed dual and showed that it captures the
visibility properties of deflated polygons.  We then showed that two
deflated polygons with the same directed dual can be monotonically
deformed into one another.  Next, we showed that directed duals can be
used to reason combinatorially, via directed dual compatibility, about
the deflatability of polygons.  Finally, we presented a polygon that
cannot be deflated\iffullversion , showed that a non-deflatable
general position polygon must have at least seven vertices\fi\ and
showed that the vertex-vertex visibilities of a polygon do not
determine its deflatability.

A full characterization of deflatable polygons still remains to be
found.  If the converse of Theorem~\ref{thm:no-cdd-not-defl} is true,
then the existence of a compatible directed dual gives such a
characterization.  We conjecture the following weaker statement.

\begin{conj}
  The vertex-edge visibilities of a polygon in general position
  determine its deflatability.
\end{conj}

We conclude, however, by noting that, in light of Mnev's Universality
Theorem \cite{Mnev88}, it is unknown if even the order type of a
polygon's vertex set determines its deflatability.

\section{Acknowledgements}

This research was partly funded by NSERC and by Carleton University
through an I-CUREUS internship.  We would also like to thank Joseph
O'Rourke for informing us of a property of point set order types.

{ \small
\bibliographystyle{abbrv}
\bibliography{poly-defl}{}
}

\iffullversion

\appendix
\section{Hexagon Enumeration}
\label{apx:hexagons}
The following figures enumerate the hexagons with two or more reflex
vertices from the case analysis of Theorem~\ref{thm:hex-deflate}.  The
enumeration includes all such hexagons on general position vertex
sets, up to order type.  A point set is in general position if no
three of its points are incident to a common line.  The order type of
a point set is a combinatorial structure that encodes, for each
ordered triple of distinct points, whether they form a right or left
turn (see \cite{Goodman83}).

The hexagons were generated by, first, enumerating the Hamiltonian
cycles, up to traversal direction, on the complete geometric graphs of
each of the sixteen general position point sets of size six from the
online order type database of Aichholzer et
al.\footnote{\url{http://www.ist.tugraz.at/aichholzer/research/rp/triangulations/ordertypes/}}
(see \cite{Aichholzer02}).  The cycles were then filtered to remove
those with edge crossings (non-simple), those without visibility
crossings (deflated) and those with less than two reflex vertices.

The dotted segments in the figures join visible vertex pairs and the
arrows indicate a vertex and a single-segment path along which it may
be moved monotonically to reduce the number of visibility crossings of
its hexagon by at least one.  The shaded region of a figure is the
convex polygon $C$, as described in the proof of
Theorem~\ref{thm:hex-deflate}.

\noindent
\includegraphics{figs/hexagons-04-012354.mps}
\includegraphics{figs/hexagons-04-013254.mps}
\includegraphics{figs/hexagons-04-014235.mps}
\includegraphics{figs/hexagons-04-014325.mps}
\includegraphics{figs/hexagons-04-041235.mps}
\includegraphics{figs/hexagons-05-012345.mps}

\noindent
\includegraphics{figs/hexagons-05-012354.mps}
\includegraphics{figs/hexagons-05-013245.mps}
\includegraphics{figs/hexagons-05-013254.mps}
\includegraphics{figs/hexagons-05-031254.mps}
\includegraphics{figs/hexagons-05-041325.mps}
\includegraphics{figs/hexagons-05-043125.mps}

\noindent
\includegraphics{figs/hexagons-06-012354.mps}
\includegraphics{figs/hexagons-06-013254.mps}
\includegraphics{figs/hexagons-06-041235.mps}
\includegraphics{figs/hexagons-06-041325.mps}
\includegraphics{figs/hexagons-07-012345.mps}
\includegraphics{figs/hexagons-07-012354.mps}

\noindent
\includegraphics{figs/hexagons-07-013245.mps}
\includegraphics{figs/hexagons-07-013254.mps}
\includegraphics{figs/hexagons-07-031254.mps}
\includegraphics{figs/hexagons-07-041235.mps}
\includegraphics{figs/hexagons-07-041325.mps}
\includegraphics{figs/hexagons-07-043125.mps}

\noindent
\includegraphics{figs/hexagons-09-012453.mps}
\includegraphics{figs/hexagons-09-012543.mps}
\includegraphics{figs/hexagons-09-013245.mps}
\includegraphics{figs/hexagons-09-013425.mps}
\includegraphics{figs/hexagons-09-031245.mps}
\includegraphics{figs/hexagons-09-031254.mps}

\noindent
\includegraphics{figs/hexagons-09-031425.mps}
\includegraphics{figs/hexagons-10-012345.mps}
\includegraphics{figs/hexagons-10-012354.mps}
\includegraphics{figs/hexagons-10-013245.mps}
\includegraphics{figs/hexagons-10-013254.mps}
\includegraphics{figs/hexagons-10-041235.mps}

\noindent
\includegraphics{figs/hexagons-10-041325.mps}
\includegraphics{figs/hexagons-10-045312.mps}
\includegraphics{figs/hexagons-10-053124.mps}
\includegraphics{figs/hexagons-10-054312.mps}
\includegraphics{figs/hexagons-11-012354.mps}
\includegraphics{figs/hexagons-11-012534.mps}

\noindent
\includegraphics{figs/hexagons-11-012543.mps}
\includegraphics{figs/hexagons-11-013254.mps}
\includegraphics{figs/hexagons-11-041235.mps}
\includegraphics{figs/hexagons-11-043512.mps}
\includegraphics{figs/hexagons-12-012345.mps}
\includegraphics{figs/hexagons-12-012354.mps}

\noindent
\includegraphics{figs/hexagons-12-012534.mps}
\includegraphics{figs/hexagons-12-012543.mps}
\includegraphics{figs/hexagons-12-013245.mps}
\includegraphics{figs/hexagons-12-013254.mps}
\includegraphics{figs/hexagons-12-015234.mps}
\includegraphics{figs/hexagons-12-015243.mps}

\noindent
\includegraphics{figs/hexagons-12-031254.mps}
\includegraphics{figs/hexagons-12-034512.mps}
\includegraphics{figs/hexagons-12-041325.mps}
\includegraphics{figs/hexagons-12-043125.mps}
\includegraphics{figs/hexagons-12-045123.mps}
\includegraphics{figs/hexagons-12-051234.mps}

\noindent
\includegraphics{figs/hexagons-13-012453.mps}
\includegraphics{figs/hexagons-13-012543.mps}
\includegraphics{figs/hexagons-13-013245.mps}
\includegraphics{figs/hexagons-13-013425.mps}
\includegraphics{figs/hexagons-13-015234.mps}
\includegraphics{figs/hexagons-13-015243.mps}

\noindent
\includegraphics{figs/hexagons-13-031245.mps}
\includegraphics{figs/hexagons-13-031254.mps}
\includegraphics{figs/hexagons-13-031425.mps}
\includegraphics{figs/hexagons-13-031524.mps}
\includegraphics{figs/hexagons-13-034125.mps}
\includegraphics{figs/hexagons-13-043125.mps}

\noindent
\includegraphics{figs/hexagons-13-051243.mps}
\includegraphics{figs/hexagons-14-012345.mps}
\includegraphics{figs/hexagons-14-012354.mps}
\includegraphics{figs/hexagons-14-012534.mps}
\includegraphics{figs/hexagons-14-012543.mps}
\includegraphics{figs/hexagons-14-013254.mps}

\noindent
\includegraphics{figs/hexagons-14-015234.mps}
\includegraphics{figs/hexagons-14-015243.mps}
\includegraphics{figs/hexagons-14-031245.mps}
\includegraphics{figs/hexagons-14-031254.mps}
\includegraphics{figs/hexagons-14-031524.mps}
\includegraphics{figs/hexagons-14-034512.mps}

\noindent
\includegraphics{figs/hexagons-14-043125.mps}
\includegraphics{figs/hexagons-14-045123.mps}
\includegraphics{figs/hexagons-14-051234.mps}
\includegraphics{figs/hexagons-15-012435.mps}
\includegraphics{figs/hexagons-15-012453.mps}
\includegraphics{figs/hexagons-15-012543.mps}

\noindent
\includegraphics{figs/hexagons-15-013245.mps}
\includegraphics{figs/hexagons-15-013425.mps}
\includegraphics{figs/hexagons-15-014235.mps}
\includegraphics{figs/hexagons-15-015234.mps}
\includegraphics{figs/hexagons-15-015243.mps}
\includegraphics{figs/hexagons-15-015342.mps}

\noindent
\includegraphics{figs/hexagons-15-015432.mps}
\includegraphics{figs/hexagons-15-031245.mps}
\includegraphics{figs/hexagons-15-034512.mps}
\includegraphics{figs/hexagons-15-045123.mps}
\includegraphics{figs/hexagons-15-051234.mps}
\includegraphics{figs/hexagons-15-051324.mps}

\noindent
\includegraphics{figs/hexagons-15-051423.mps}
\includegraphics{figs/hexagons-15-051432.mps}
\includegraphics{figs/hexagons-15-054123.mps}
\includegraphics{figs/hexagons-16-012345.mps}
\includegraphics{figs/hexagons-16-012354.mps}
\includegraphics{figs/hexagons-16-013542.mps}

\noindent
\includegraphics{figs/hexagons-16-015342.mps}
\includegraphics{figs/hexagons-16-043512.mps}
\includegraphics{figs/hexagons-16-045132.mps}
\includegraphics{figs/hexagons-16-045312.mps}
\includegraphics{figs/hexagons-16-051324.mps}
\includegraphics{figs/hexagons-16-053124.mps}

\noindent
\includegraphics{figs/hexagons-16-054312.mps}
\fi

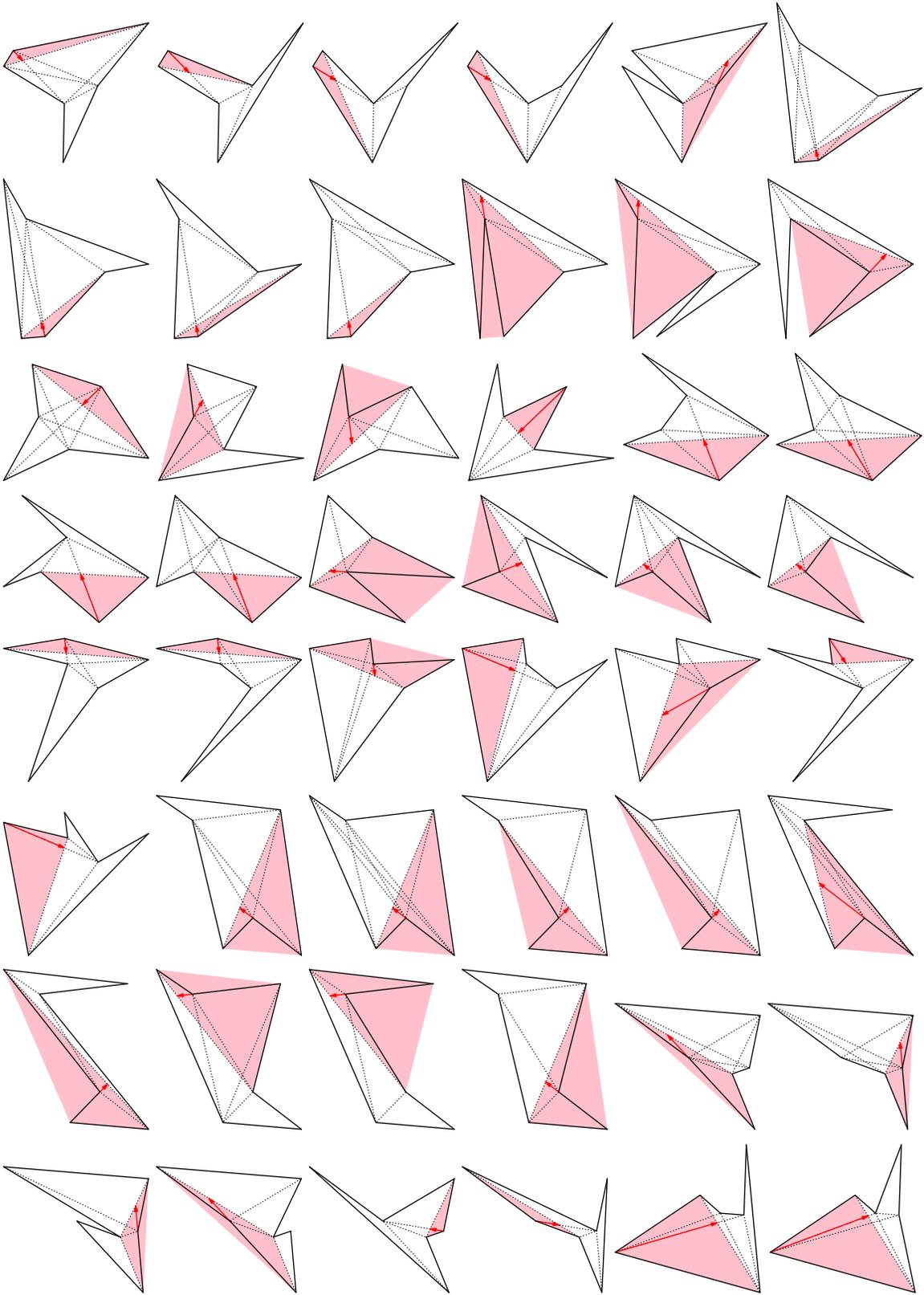
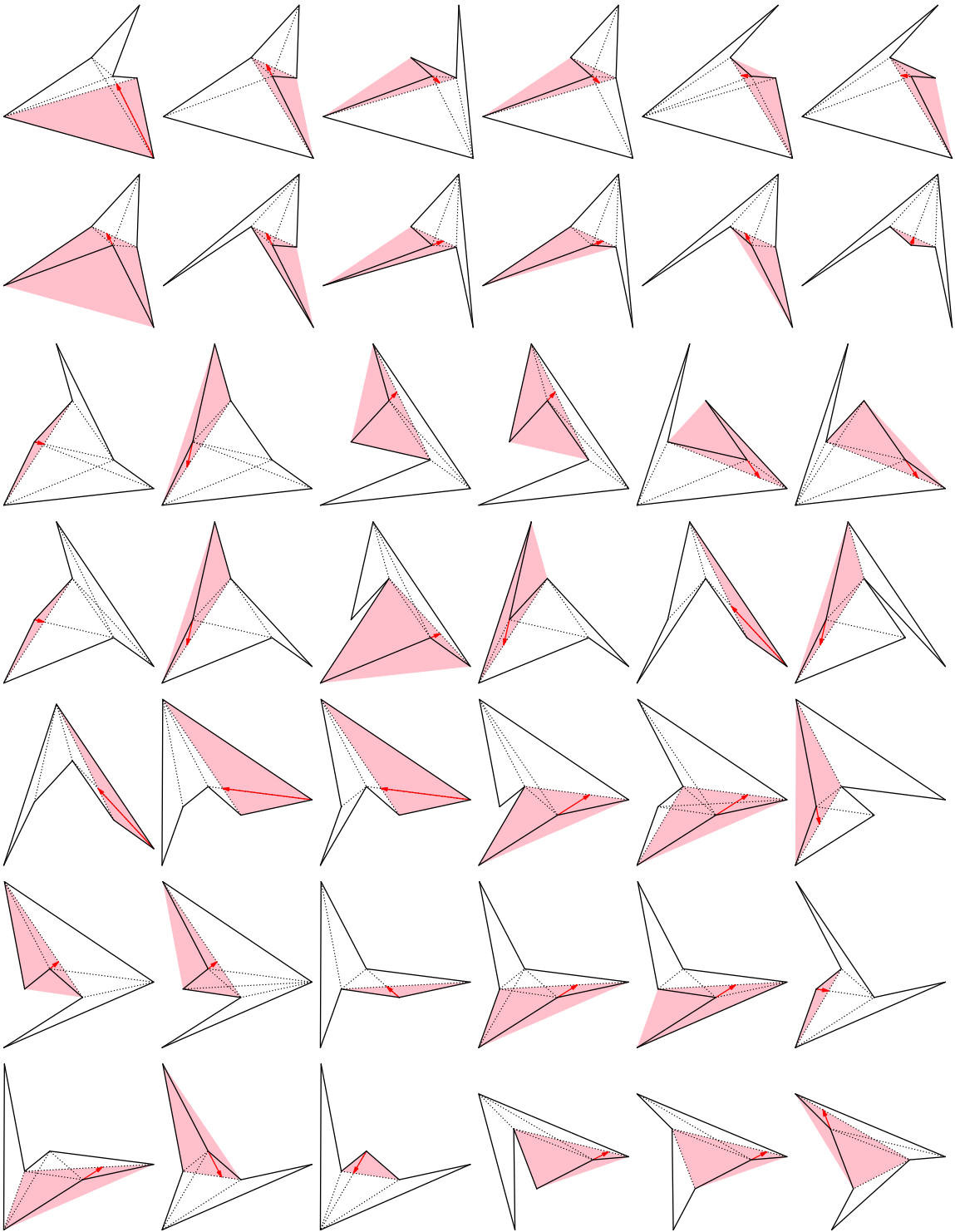
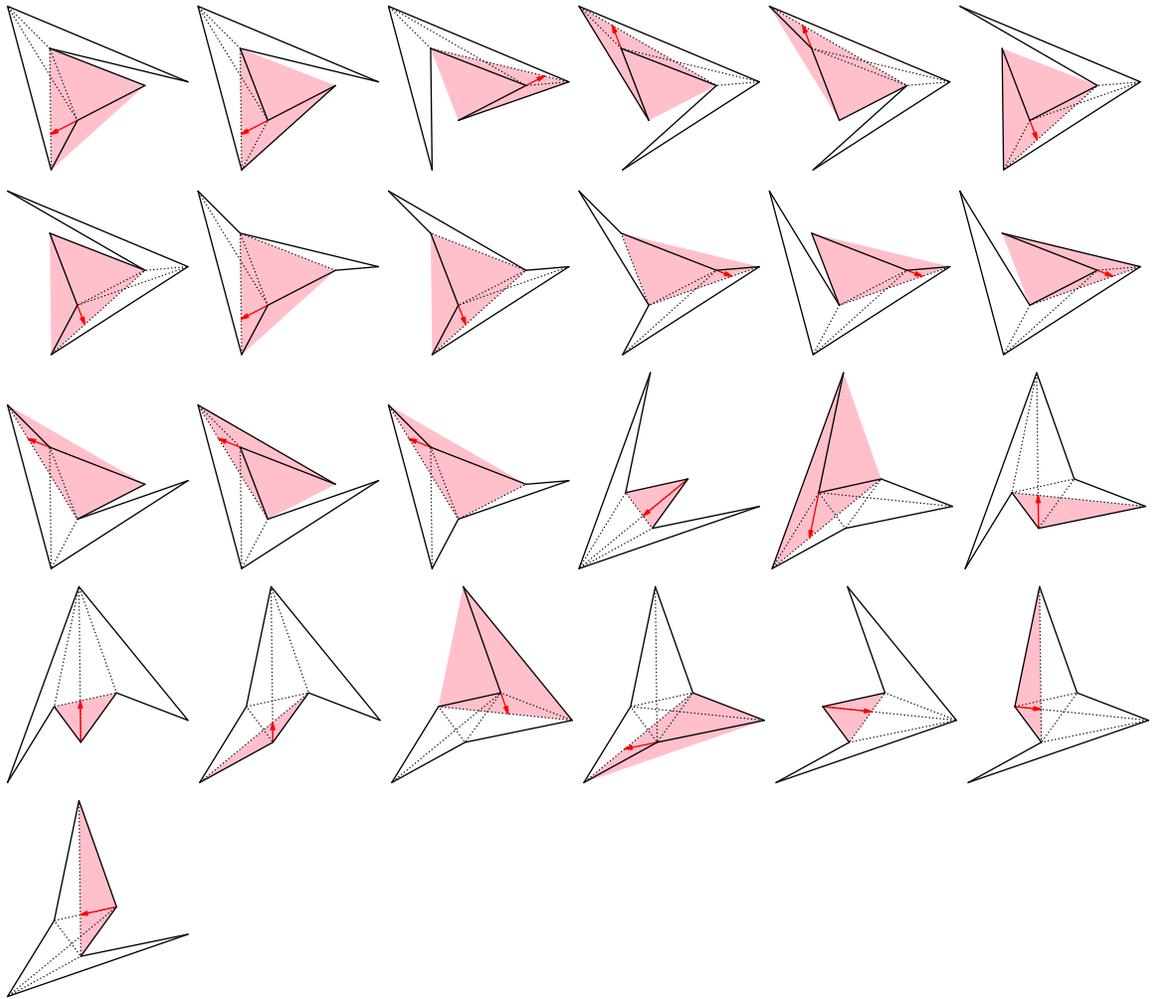
\end{document}